\documentclass{article}
\usepackage[english]{babel}
\usepackage[T1]{fontenc}
\usepackage[utf8]{inputenc}
\usepackage{amsthm}
\usepackage{amsmath}
\usepackage{amssymb}
\usepackage{palatino}
\usepackage{graphicx}
\usepackage[colorlinks=true, citecolor =blue]{hyperref}
\usepackage{dsfont}
\usepackage[all]{xy}
\usepackage[top=4cm,bottom=4cm,left=3cm,right=3cm]{geometry}
\usepackage{enumitem}
\usepackage{mathtools}
\usepackage{subcaption}
\usepackage{authblk}
\usepackage{colortbl}
\usepackage{comment}
\usepackage{empheq}

\newcommand{\R}{\mathbb{R}}

\theoremstyle{plain}

\newtheorem{thm}{Theorem}

\newtheorem{prop}{Proposition}

\newtheorem{rmk}{Remark}

\title{A qualitative analysis of a A$\beta$-monomer model with inflammation processes for Alzheimer's disease}
\date{November 2022}

\author[a]{Ionel Ciuperca}
\author[b]{Laurent Pujo-Menjouet}
\author[b]{Leon Matar-Tine}
\author[b]{Nicolas Torres\thanks{Corresponding author. Email : torres@math.univ-lyon1.fr}}
\author[a,c]{Vitaly Volpert}
\affil[a]{Univ Lyon, Université Claude Bernard Lyon 1, CNRS UMR 5208, Institut Camille Jordan, F-69622 Villeurbanne, France.}
\affil[b]{Université Claude Bernard Lyon 1, CNRS UMR 5208, Institut Camille Jordan, Inria DRACULA, F-69603 Villeurbanne, France.}
\affil[c]{Peoples’ Friendship University of Russia, 6 Miklukho-Maklaya St, Moscow, 117198, Russia.}

\usepackage{graphicx}

\begin{document}

\maketitle

\begin{abstract}
We introduce and study a new model for the progression of Alzheimer's disease incorporating the interactions of A$\beta$-monomers, oligomers, microglial cells and interleukins with neurons through different mechanisms such as protein polymerization, inflammation processes and neural stress reactions. In order to understand the complete interactions between these elements, we study a spatially-homogeneous simplified model that allows to determine the effect of key parameters such as degradation rates in the asymptotic behavior of the system and the stability of equilibriums. We observe that inflammation appears to be a crucial factor in the initiation and progression of Alzheimer's disease through a phenomenon of hysteresis, which means that there exists a critical threshold of initial concentration of interleukins that determines if the disease persists or not in the long term. These results give perspectives on possible anti-inflammatory treatments that could be applied to mitigate the progression of Alzheimer's disease. We also present numerical simulations that allow to observe the effect of initial inflammation and concentration of monomers in our model.
\end{abstract}

\textit{Keywords:} Alzheimer's Disease; Persistence; Bifurcation Analysis; Hysteresis; Inflammation.

\section{Introduction}
Understanding the origin and development of Alzheimer's disease (AD) has been a challenging problem for biologists during the past decades. As in many neurodegenerative diseases, AD is known to be associated with the misconformation, aggregation and propagation of
different proteins in the neural system \cite{haass2007soluble,Sakono2010,sengupta2016role,soto2003unfolding}. They form stable oligomers that eventually accumulate in the so called amyloid plaques and this phenomenon is believed to lead to a progressive irreversible neuronal damage. One of these proteins that appears to be relevant in the early stages of development of AD are the A$\beta$-monomers, whose precise mechanisms of aggregation and diffusion are yet to be discovered.

In this context, mathematical models arise as a useful approach to understand the different processes underlying Alzheimer. Several types of models have been considered, including from simple systems of ordinary differential equations to more complex partial differential equations such as transport equations \cite{ciuperca2019alzheimer}, reaction-diffusion models \cite{matthaus2006diffusion,matthaeus2009spread,Bertsch2016,andrade2020modeling} and stochastic control models \cite{hu2022finite}. 

The goal of this article is to understand the complete interactions between A$\beta$-monomers, oligomers, microglial cells and interleukins through a new system of partial differential equations, involving the development of AD in the brain. Neurons produce A$\beta$-monomers, that almost instantaneously start to polymerize into proto-oligomers. In this aggregation process proto-oligomers are able to polymerize or depolymerize and once they reach a critical size they become stable under the form of A$\beta$-oligomers. These latter are assumed to be totally stable in the sense that neither polymerization nor depolymerization  is possible for A$\beta$-oligomers equilibrium \cite{murphy2000probing,nag2011nature}.

Besides the mechanism of polymerization, oligomers interact with microglial cells, considered as auxiliary cells in the nervous systems regulating brain development. They induce an inflammation reaction through a chemical cascade in microglial cells, releasing interleukins \cite{forloni2018alzheimer,kinney2018inflammation}. These interleukins then activate an increase of A$\beta$-monomers production from the neurons. However, if the concentration of A$\beta$-oligomers is high enough, then a reaction of stress called UPR (Unfolded protein response) \cite{soto2003unfolding} is triggered which leads to a decrease of A$\beta$-monomers production, while the rest of oligomers diffuses in the neuronal environment. In this context, two opposed mechanisms of stimulation and inhibition will determine the persistence of AD or not.

Moreover, oligomers are brought and displaced by microglia to the amyloid plaques, i.e. an aggregate of A$\beta$-oligomers that becomes an inert element (no diffusion, no polymerization, no depolymerization). Each element of the system (monomers, proto-oligomers, and oligomers except those in the amyloid plaques) diffuses, with a size-dependent rate. Microglial cells can also have random motility, but they displace free A$\beta$-oligomers to the amyloid plaques through a chemotactic process and amyloid plaques will more likely develop where the concentration of microglial cells is high. These cells are known indeed to be very reactive to neuronal insults \cite{hansen2018microglia,mazaheri2017trem,ransohoff2016polarizing,sarlus2017microglia}. 

Inflammation processes seem to be crucial to control the disease progression \cite{kinney2018inflammation} and to find possible therapeutic strategies to mitigate negative effects of AD. For example, it suggested in \cite{rivers2020anti} that diclofenac might be associated with slower cognitive decline with possible perspectives on AD progression. However, despite epidemiological evidences, robust clinical trials have not been successful in providing efficacy evidence of such anti-inflammatory treatments \cite{adapt2008cognitive,adapt2007naproxen,ozben2019neuro}. On the other hand, in \cite{ali2019recommendations,imbimbo2010nsaids} it is suggested that anti-inflammatory treatments might be effective if they are applied years before the development of clinical symptoms. Furthermore in \cite{imbimbo2010nsaids}, it is mentioned that some anti-inflammatory treatments decrease the levels of A$\beta$ by
allosterically inhibiting the $\gamma$-secretase complex, which could give interesting perspectives in finding efficient cures. Other treatment suggestions include actions on multiple targets besides neuroinflammatory and neuroprotective effects such as anti-amyloid and anti-tau effects \cite{huang2020clinical,zhu2018can}.

The article is organized as follows. In Section \ref{mathmodel} we introduce the main system of partial differential equations and we describe the reactions involving monomers, (proto-)oligomers, microgial cells and interleukins, which are summarized in Figure \ref{fig:inflammatory}. This model incorporating spatial dependence and aggregation processes is inspired in previous works such as \cite{Bertsch2016,andrade2020modeling} and it will serve as our base model. As a general goal, we aim to understand the progression of AD through an analysis compatible simplified version for this base model. In particular, in Section \ref{bimonomeric} we focus on a spatially-homogeneous version of the main model, where polymerization and de-polymerization processes is simplified. For this simplified model we analyze the existence of steady states depending on the parameters. Finally in Section \ref{numericalsim}, we present numerical simulations of the simplified model in order to observe the different possible dynamics of solutions and the stability of the steady states.

\begin{figure}[!ht]
\centering
\includegraphics[scale=0.35]{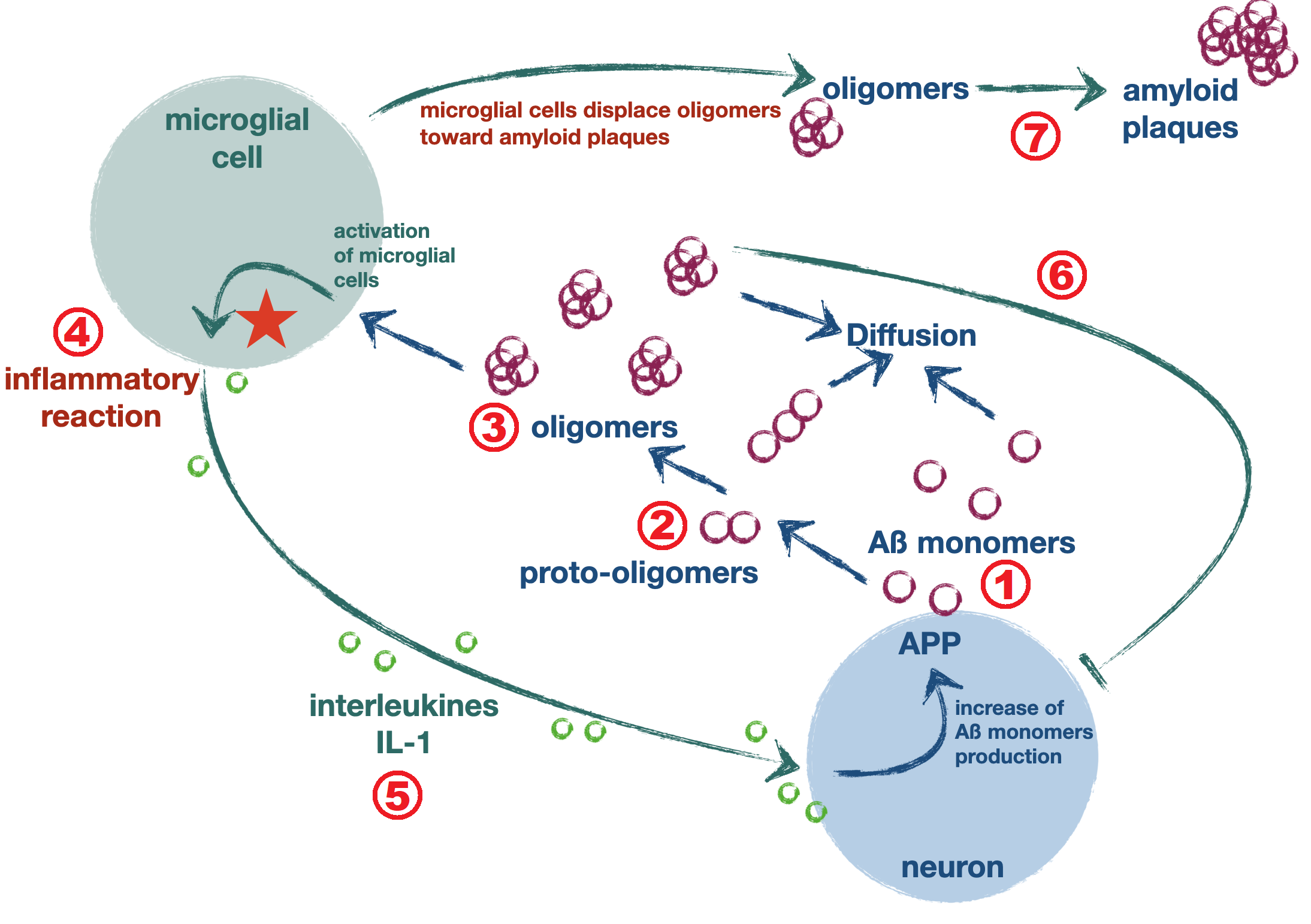}
\caption{{\bf Schematic representation of A$\beta$-monomers and inflammation cycle}. Neurons produce A$\beta$-monomers \textbf{(1)} that polymerize into  proto-oligomers \textbf{(2)}. These proto-oligomers eventually reach a critical size to become stable oligomers \textbf{(3)}. They activate microglial cells triggering an inflammatory reaction \textbf{(4)} by producing interlekins. The interleukins stimulate neurons \textbf{(5)} to increase A$\beta$-monomers production, closing the positive feedback cycle. Moreover when oligomer concentration is high, neurons are stressed \textbf{(6)} and decrease the A$\beta$-monomers production, while oligomers are displaced by microglial cells toward the amyloid plaques \textbf{(7)}.}
\label{fig:inflammatory}
\end{figure}

\section{Mathematical model}
\label{mathmodel}

Let us detail each equation of the system. In this model, we consider that dynamics occur in a part of the brain considered as an open bounded domain $\Omega\subset\R^d$ (with $d\in\{2,3\}$) and the main variables of the system are given in Table \ref{tab:variables}.

\begin{table}[!ht]
    \centering
    \caption{\bf Variables of the mathematical model}
    \begin{tabular}{|l|l|}\hline
        Variable & Definition \\ \hline
        $u_i(t,x)$ & Concentration of A$\beta$-proto-oligomers of size $i$. \\ \hline $u(t,x)$ & Concentration of A$\beta$-oligomers.\\ \hline
        $u_p(t,x)$ & Concentration of oligomers in the amyloid plaques.\\ \hline
        $m(t,x)$  & Concentration of A$\beta$-monomers. \\ \hline
        $M(t,x)$ & Concentration of microglial cells.\\ \hline
        $I(t,x)$ & Concentration of interleukins.\\ \hline
    \end{tabular}
    \label{tab:variables}
\end{table}

\begin{enumerate}
     \item \textbf{Proto-oligomers}: (see point \textbf{(2)} in Figure \ref{fig:inflammatory}). A$\beta$-proto-oligomers have a size ranging from $i=2$ to $i=i_0-1$ and become oligomers when they reach the size $i=i_0$. Equations for proto-oligomers with size $i=[2,\cdots,i_0-1]$ are given by
   
    $$\dfrac{\partial u_i}{\partial t} (t,x) =  r_{i-1}u_{i-1}(t,x)m(t,x)+b_i u_{i+1}(t,x)
             -r_i u_i(t,x)m(t,x)
     -b u_i(t,x)  + \nu_i \Delta u_i(t,x),$$

    with $r_1$ the bi-monomeric nucleation rate, $u_1= m(t,x)$ the monomer concentration , $b_i$ the rate of monomer loss from proto-oligomers and $r_i$ the rate of monomer gain. The couple $(r_i,b_i)$, $i\in [2,\cdots,i_0-1]$ is called kinetic coefficients with the notation $b_i=b$ if $i\leq i_0-2$ and $b_{i_0 - 1} = 0$.
    
    The first term of the right-hand side stands in one hand for the bi-monomeric nucleation when $i=2$ and on the other hand for the polymerization with rate $r_{i-1}$ ($i\geq 3$) of a proto-oligomer of size $i-1$ with the contact of a monomer giving then a proto-oligomer of size $i$.  The second term describes the depolymerization with rate $b_i$ of a proto-oligomer of size $i+1$ to a proto-oligomer of size $i$. The third and fourth term are related to the symmetric 
    process respectively of polymerization and depolymerization of a proto-oligomer of size $i$.  Finally each proto-oligomer can diffuse with a size dependent coefficient (the smaller the size the faster the diffusion). 

    \item \textbf{Free oligomers}: (see point \textbf{(3)} in Figure \ref{fig:inflammatory}). The variation of the A$\beta$-oligomer population is described as follows
     \begin{center}
    $\begin{array}{lll}
            \dfrac{\partial u}{\partial t} (t,x)  &  =  &  r_{i_{0}-1}u_{i_{0}-1}(t,x)m(t,x)-\gamma(M(t,x))u(t,x) -\tau_0 u(t,x) +  \nu_{i_0} \Delta u(t,x),
    \end{array}$
    \end{center}

     where the first term of the right-hand side stands for the polymerization with rate $r_{i_0-1}$ of a proto-oligomer of size $i_0-1$ with the contact of a monomer giving then an oligomer of size $i_0$. The second term describes the recruitment of oligomers to the amyloid plaques by microglial cells $M$ with a rate $\gamma$ given by
     
     $$\gamma(M)=\gamma_0+\dfrac{\gamma_1M}{1+\gamma_2M},$$
     
     depending on $M$ through a Michaelis-Menten function with parameters $\gamma_i$ ($i\in \{0,1,2\}$) and the third term corresponds to the degradation of oligomers with rate $\tau_0$. Finally each oligomer diffuses with rate $\nu_{i_0}$. It is important to remind here that oligomers neither polymerize nor depolymerize unlike proto-oligomers.
     
    \item \textbf{Oligomers in the amyloid plaques}: (see point \textbf{(7)} in Figure \ref{fig:inflammatory}). The variation of the A$\beta$-oligomer population stuck in the amyloid plaques is described as follows
     \begin{center}
    $\begin{array}{lll}
            \dfrac{\partial u_p}{\partial t} (t,x)  &  =  &   \gamma(M(t,x))u(t,x)-\tau_p u_p(x,t),
    \end{array}$
    \end{center}
     where the first term of the right-hand side stands for the recruitment of free oligomers to the amyloid plaques by microglial cells $M$ with a rate $\gamma$ and the second term represents the corresponding loss with rate $\tau_p$. We remind here that  oligomers in the amyloid plaques neither polymerize, depolymerize nor diffuse.

      \item \textbf{Monomers}: (see point \textbf{(1)} in Figure \ref{fig:inflammatory}). The variation of the A$\beta$-monomer population is described as follows
     \begin{center}
    $\begin{array}{lll}
            \dfrac{\partial m}{\partial t} (t,x)  &  =  
            &  - r_1 m^2  -\displaystyle \sum_{i=2}^{i_0-1} r_i u_i(t,x)m(t,x) +b \displaystyle \sum_{i=2}^{i_0-1}  u_i(t,x)\\
            & & + S( u(t,x), I(t,x))-d\, m(t,x)+\nu_{1} \Delta m(t,x),
    \end{array}$
    \end{center}
    where $I(t,x)$ is the concentration of interleukins and the function $S$ is given by
    \begin{equation}
    \label{stress}
        S(u,I) = \frac{\tau_S}{1 + Cu^n} I,\qquad n\ge 1.
    \end{equation}
 The term $S(u,I)$ is called the stress function. According to the form of this function, under a high concentration of oligomers $u$ surrounding the neuron, this latter will be stressed and stop the production of A$\beta$ monomers, which means that $S(u,I)$ is close to zero (see point \textbf{(6)} in Figure \ref{fig:inflammatory}).  \\
We remark that the neuron can be torn between the decision of producing A$\beta$-monomers due to the inflammation (caused by
the interleukins) and the stress caused by the amount of oligomers surrounding the neurons causing the UPR process that stops this A$\beta$ production. Note that this object is one the major key properties in our model. For simplicity, we do not take into account the fact that microglia produce A$\beta$-monomers and this will be considered in a future work with a more complex model.

The first and second terms of the right-hand side stand respectively for the bi-monomeric nucleation and the polymerization of proto-oligomers of all sizes, while the third term describes the corresponding processes of depolymerization of proto-oligomers. The fourth term is the source term depending on the inflammation reaction caused by interaction of A$\beta$-oligomers with microglial cells. The fifth term describes the degradation of the monomers with a rate $d$. This rate $d$ may depend on oligomer concentration and behave as a Hill function, but for simplicity we consider in the sequel that $d$ is a given positive constant. Finally, the last term stands for the monomer diffusion ability with rate $\nu_1$.

\item \textbf{Microglial cells}: (see point \textbf{(4)} in Figure \ref{fig:inflammatory}). The evolution of the microglial cells population is described as follows
     \begin{center}
    $\begin{array}{lll}
            \dfrac{\partial M}{\partial t} (t,x)  &  =  &   D_1\Delta M(t,x) -\alpha \nabla\cdot(M(t,x)\nabla u(t,x))\\
            & &+\lambda_M+\dfrac{\alpha_1u(t,x)}{1+\alpha_2u(t,x)}\left(\hat{M}-M(t,x)\right)M(t,x) - \sigma M(t,x),
    \end{array}$
    \end{center}
    where the first term of the right-hand side stands for the diffusion of microglial cells with the rate $D_1$. The second term represents the chemotaxis of microglial cells in response to the increase of oligomers population. This chemotactic effect results in an activation of microglial cells due to the presence of oligomers which causes an inflammatory reaction with production of interleukins (IL-1). The third term describes the proliferation of microglial cells with a constant rate $\lambda_M$. In the fourth term $\hat{M}$ is the maximum capacity of microglial cells in the neuron environment and the last term characterizes the loss of microglial cells with the rate $\sigma$.\\

\item \textbf{Interleukins}: (see point \textbf{(5)} in Figure \ref{fig:inflammatory}). The equation for the evolution of interleukins is:
\begin{center}
    $$ \dfrac{\partial I}{\partial t} (t,x)    =     D_I \Delta I(t,x) + \dfrac{\tau_1 u(t,x)}{1 + \tau_2 u(t,x)}M(t,x)  - \tau_3 I(t,x),$$
    \end{center}
where the first term in the right-hand side is the diffusion
of the interleukins, the second term represents the proliferation which depends on the concentration of oligomers through a Michaelis-Menten function with parameters $\tau_1,\tau_2$ and the microglial cells. The third term represents the loss of interleukins with rate $\tau_3$.
\end{enumerate}

We note that all equations are complemented with Neumann boundary conditions with zero flux through $\partial\Omega$ and the parameters of the system are non-negative real numbers. The main interactions of this system are summarized in Figure \ref{fig:inflammatory}.


\section{A bi-monomeric simplified model}
\label{bimonomeric}

In order to proceed to a full mathematical analysis, and understand the qualitative dynamics of the actors of this problem, we consider a simplified model version of the full system of partial differential equations. We assume a bi-monomeric nucleation, \textit{i.e.} two monomers can merge to form a free oligomer ($m+m\to u$) and the intermediate proto-oligomer phase is absent. For this case, we assume that when a monomer attaches to a free oligomer, the latter does not change and the monomer is consumed ($u+m\to u$). The equations of the simplified PDE system are the following

\begin{empheq}[left=\empheqlbrace]{align}
    \frac{\partial u}{\partial t}=&\nu_{2} \Delta u  +  r_1 m^2-\gamma(M)u-\tau_0 u,\nonumber\\
    \frac{\partial u_p}{\partial t}=&\gamma(M)u-\tau_p u_p,\nonumber \\
    \frac{\partial m}{\partial t}=&\nu_{1} \Delta m + \frac{\tau_S}{1+Cu^n}I-dm- r_2 um-r_1 m^2,\label{edp-simple}\\
    \frac{\partial M}{\partial t}=& D_1\Delta M -\alpha \nabla\cdot(M\nabla u)+\frac{\alpha_1 u}{1+\alpha_2 u}(\hat{M}-M)M-\sigma M +\lambda_M,\nonumber\\
    \frac{\partial I}{\partial t}=&D_I \Delta I+\frac{\tau_1 u}{1+\tau_2 u}M-\tau_3 I.\nonumber
\end{empheq}

We also assume that when a monomer binds an oligomer, then the monomer is consumed with rate $r_2$ and the number oligomer molecules does not change. Under these assumptions, we notice that there is no term involving the rate $r_2$ in the equation of oligomers.

\subsection{Spatially homogeneous model}

To simplify the analysis in this work, we focus on spatially-homogeneous solutions of the bi-monomeric model \eqref{edp-simple}. For simplicity, we assume that rate of recruitment of oligomers to the amyloid plaques $\gamma(M)$ is constant, which corresponds essentially to consider an average rate of oligomers being recruited and we consider that oligomers have a highly stable structure and their degradation is negligible, which means $\tau_0=0$. However, the results on the qualitative analysis of the system do not change if we consider the degradation of oligomers. Under this setting, the model is reduced to the following system of ordinary differential equations

\begin{empheq}[left=\empheqlbrace]{align}
    \frac{du}{dt}=&r_1 m^2 - \gamma_0 u,\nonumber\\
    \frac{du_p}{dt}=&\gamma_0 u-\tau_p u_p,\nonumber \\
    \frac{dm}{dt}=& \frac{\tau_S}{1+Cu^n}I-dm-r_2 um-r_1 m^2,\label{ode1-u-m-M-I}\\
    \frac{dM}{dt}=&\frac{\alpha_1 u}{1+\alpha_2 u}(\hat{M}-M)M-\sigma M +\lambda_M,\nonumber\\
    \frac{dI}{dt}=&\frac{\tau_1 u}{1+\tau_2 u}M-\tau_3 I.\nonumber
\end{empheq}

Thanks to this simplification we obtain the following result.

\begin{prop}
For any non-negative initial condition $(u(0),u_p(0),m(0),M(0),I(0))$, the system has a unique global solution which is bounded.
\end{prop}

\begin{proof}
Existence and uniqueness of a local solution is straightforward from Cauchy-Lipschitz Theorem for ordinary differential equations. Positivity of solutions comes from the standard quasi-positivity argument (see for instance \cite{haraux2016simple}, Section 2) and boundedness is obtained analogously. Since solutions of system \eqref{ode1-u-m-M-I} are bounded, they are defined for all $t>0$.
\end{proof}

\subsection{Steady states}

The stationary points of system \eqref{ode1-u-m-M-I}, correspond to solutions of the following system
\begin{empheq}[left=\empheqlbrace]{align}
   r_1 m^2 - \gamma_0 u=0,\nonumber\\
   \gamma_0 u-\tau_p u_p=0,\nonumber \\
   \frac{\tau_S}{1+Cu^n}I-dm-r_2 um-r_1 m^2=0,\label{ode1-u-m-M-I-est}\\
    \frac{\alpha_1 u}{1+\alpha_2 u}(\hat{M}-M)M-\sigma M +\lambda_M=0,\nonumber\\
    \frac{\tau_1 u}{1+\tau_2 u}M-\tau_3 I=0.\nonumber
\end{empheq}

One of the solutions of this system is the disease-free equilibrium, given by $\left(0,0,0,\tfrac{\lambda_M}{\sigma},0\right)$. Besides this equilibrium there may be others steady states depending on the parameter values of our system, whose existence will be studied in this section.
Concerning the disease-free equilibrium, we get the following result.

\begin{prop}
For the system \eqref{ode1-u-m-M-I}, the disease-free equilibrium $\left(0,0,0,\tfrac{\lambda_M}{\sigma},0\right)$ is locally asymptotically stable for every choice of positive parameters.
\end{prop}

\begin{proof}
The Jacobian matrix around the vector $\left(0,0,0,\tfrac{\lambda_M}{\sigma},0\right)$ is given by
\begin{equation*}
    J=\begin{bmatrix}
    -\gamma_0 & 0 & 0 & 0 & 0\\
    \gamma_0 & -\tau_p & 0 & 0 & 0\\
    0 & 0 & -d & 0 & \tau_S\\
    \alpha_1\left(\hat{M}-\frac{\lambda}{\sigma}\right)\frac{\lambda}{\sigma} & 0 & 0 & -\sigma &0\\
    0 & 0 & 0 & 0 &-\tau_3
    \end{bmatrix},
\end{equation*}
whose set of eigenvalues is given by $\{-\gamma_0,-\tau_p,-d,-\sigma,-\tau_3\}$. Since they are all negative, then the disease-free equilibrium is locally asymptotically stable.
\end{proof}

An interesting question is to determine under which parameter values existence of non-trivial steady-states (\textit{i.e.} AD persists) holds. In this regard, we have the following result.

\begin{thm}
\label{existence}
Assume that the parameters satisfy the condition
\begin{equation}
\label{cond-sigma-g0-tau3}
    \sigma\gamma_0\tau_3< \tau_1 \tau_S \lambda_M.
\end{equation}

Then for $d>0$ small enough, there exist at least two positive steady states of system \eqref{ode1-u-m-M-I-est}. If $d>0$ is large enough, then there are not positive solutions of system \eqref{ode1-u-m-M-I-est}, regardless Condition \eqref{cond-sigma-g0-tau3}.
\end{thm}

\begin{proof}
From system \eqref{ode1-u-m-M-I-est}, we solve for $u$ and $u_p$ in terms of $m$ and we get the following relation

$$u = \rho m^2,\: u_p=\frac{r_1}{\tau_p}m^2\qquad\textrm{with}\quad\rho=\frac{r_1}{\gamma_0}.$$

From the equation of microglial cells, we solve the quadratic equation of $M$ in terms of $u$ and by taking the positive root we get the following equality

\begin{equation}
\label{M-of-u-gen}
M = \frac {\sqrt{\Delta(u)} - \sigma - (\sigma \alpha_2 - \hat{M}
\alpha_1) u } {2 \alpha_1 u},
\end{equation}
with 
$\Delta(u) = (\sigma + (\sigma \alpha_2 - \hat{M} \alpha_1) u)^2 +
4 \lambda_M \alpha_1 u (1 + \alpha_2 u)$. For the interleukins we get relation
$$ I = \frac{\tau_1}{\tau_3}\frac{\rho m^2}{1 + \tau_2 \rho m^2}M. $$
Substituting these expressions into equation of $m$ in \eqref{ode1-u-m-M-I-est}, we obtain the equation with respect to $m$:
\begin{equation}
\label{eq-P-F}
m (P(m)+d) = m F(m) ,
\end{equation}
where the functions $P$ and $F$ are given by

\begin{align}
P(m) &=  r_2\rho m^2 + r_1 m, \nonumber \\
F(m) &= \frac{2 \tau_1 \tau_S\lambda_M}{\tau_3} \frac {\rho m (1 + \alpha_2 \rho m^2)}
{[\sqrt{\Delta(\rho m^2)} + \sigma +
(\sigma \alpha_2 - \hat{M} \alpha_1) \rho m^2]
(1 + \tau_2 \rho m^2) (1 + C \rho^n m^{2n})}.
\end{align}

The disease-free equilibrium corresponds to the case when $m=0$ in Equation \eqref{eq-P-F}. In order to get a positive steady state of system \eqref{ode1-u-m-M-I} we must determine the values where $P(m)+d=F(m)$.

From the definition of $\Delta(u)$, we remark that the denominator is strictly positive in the function $F$. We observe that $F(0)=0$, $F(m)>0$ for $m>0$ and $F(m) \to 0$ as $m \to \infty$, since the numerator is of order $O(m^3)$ and the denominator is of order $O(m^{2n+4})$.

Moreover $F'(0)$ is given by:
$$F'(0)= \frac{ r_1 \tau_1 \tau_S \lambda_M}{\sigma\gamma_0\tau_3}
>0.$$

From Hypothesis \eqref{cond-sigma-g0-tau3} we observe that
$$P'(0)<F'(0),$$
hence there exists $\tilde{m} > 0$ such that
\begin{equation}
\label{ineq-tilde-m}   
P(m) < F(m) \quad \textrm{for all}\: m \in (0, \tilde{m}).
\end{equation}
Let us denote 
$ m_0 = \sup \; \{\tilde{m} > 0 \colon\;\textrm{property \eqref{ineq-tilde-m} holds}\}>0$. Since $P(m)\to\infty$ as $m\to\infty$ we conclude that $m_0<\infty $ and from continuity we get 
\begin{equation}
    \label{ineq-m0}
    P(m) < F(m)\quad \textrm{for all}\: m \in (0, m_0),\qquad P(m_0) = F(m_0).
\end{equation}

Let us now denote
$$\tilde{d} = \max_{y \in [0, m_0]} (F(y) - P(y)),$$
which is strictly positive by condition \eqref{ineq-m0}. Let $y_0 \in(0, m_0)$ such that $F(y_0) - P(y_0) = \tilde{d}$. We now take an arbitrary $d$ such that $0 < d < \tilde{d}$. And the following inequalities hold

$$P(0)+d > F(0),\qquad P(y_0)+d < F(y_0),\qquad P(m_0)+d>F(m_0).$$

Therefore, there exists a positive solution of Equation \eqref{eq-P-F} in $(0,y_0)$ and another positive solution in $(y_0,m_0)$. This proves the existence result.

For the non-existence result, observe that $F$ reaches a maximum, since the $F(0)=0$ and $F(m)\to 0$ as $m\to\infty$, and this maximum is independent of $d$. Hence for $d$ large enough, we have that
$$P(m)+d>\max_{y>0}{F(y)}\ge F(m)\quad\textrm{for all}\: m\ge0,$$
and we conclude that there is no solution in that case.
\end{proof}

From the previous result we assert that when the rest of the parameters are fixed, there exists a critical value of degradation rate of monomers $d=d_c$, such that for $d>d_c$ the system \eqref{ode1-u-m-M-I} has only the disease-free equilibrium and for $d>d_c$ there are at least two positive solutions. From a biological point of view, this means that a high degradation of monomers can avoid the persistence of AD, while a lower degradation of monomers is not sufficient to stop the pathogenic cycle of monomers, oligomers and interleukins.

\newpage

\section{Numerical Simulations}
\label{numericalsim}

The main goal of this section is to present a qualitative analysis of the possible asymptotic behaviors and the stability of steady states of system \eqref{ode1-u-m-M-I} through a bifurcation diagram with respect to degradation rate of monomers and the concentration of interleukins at equilibrium. For these simulations we rely on the parameter values given in Table \ref{Table 2.}.

\begin{table}[!ht]
    \centering
    \caption{Parameter values for the numerical simulations of Equation \eqref{ode1-u-m-M-I}.}
    \begin{tabular}{|l|l|l|l|}
    \hline
        Parameter &  Value & Units & Description\\ \hline
        $r_1$ & $10^{-1}$ & $L\,(\text{mol})^{-1}(\text{months})^{-1}$ & Bi-monomeric polymerization rate\\ \hline
        $r_2$ & $10^{-1}$ & $L\,(\text{mol})^{-1}(\text{months})^{-1}$ & Polymerization rate of monomers attaching to oligomers\\ \hline
        $d$ & Variable & $(\text{month})^{-1}$ & Degradation rate of monomers\\ \hline
        $\gamma_0$ & $5\times 10^{-2}$ & $(\text{month})^{-1}$ & Recruitment rate of oligomers to the amyloid plaques\\ \hline
        $\tau_1$ & 1 & $L\,(\textrm{mol})^{-1}(\text{months})^{-1}$  &  Growth coefficient of interleukins\\ \hline
        $\tau_2$ & 1 & $L\,(\textrm{mol})^{-1}$& Growth coefficient of interleukins\\  \hline
        $\tau_3$ & $1$ & $(\text{months})^{-1}$ & Degradation rate of interleukins\\ \hline
        $\tau_p$ & $3\times 10^{-2}$& $(\text{months})^{-1}$ & Degradation rate of oligomers in the amyloid plaques \\ \hline
        $\tau_S$ & 1 & $(\text{months})^{-1}$ & Coefficient of neural stress\\ \hline
        $C$ & 1 & $L^n\,(\text{mol})^{-n}$ & Coefficient of stress function\\ \hline
        $n$ & 2 & - & Power coefficient of stress function\\ \hline
        $\alpha_1$ & 1 & $L^2\,(\text{mol})^{-2}(\text{months})^{-1}$ & Growth coefficient of microglial cells\\ \hline
        $\alpha_2$ & 1 & $L\,(\text{mol})^{-1}$ & Growth coefficient of microglial cells \\ \hline
        $\lambda_M$ & $10^{-3}$ & $\text{mol}\,L^{-1}(\text{months})^{-1}$  & Rate of proliferation of microglial cells\\ \hline
        $\hat{M}$ & $1$ & $\text{mol}\,L^{-1}$ & Capacity of microglial cells\\ \hline
        $\sigma$ & $10^{-3}$ & $(\text{months})^{-1}$ & Degration rate of microglial cells\\ \hline
    \end{tabular}
      \caption*{The values are chosen with the order of magnitude between $10^{-3}$ and $1$, in the typical range of a biological process. The values can be re-scaled if needed, but the qualitative behavior is similar. In particular, we assumed that polymerization process of monomers is faster than the corresponding degradation of monomers and oligomers.}
    \label{Table 2.}
\end{table}

From the previous analysis of Section \ref{bimonomeric}, the key parameters in determining the existence of positive steady states where the disease persists are the degradation rates. Obtaining the precise value of all parameters is in general a complicated task. However, since we are interested in the qualitative behavior of system \eqref{ode1-u-m-M-I}, modifying the values in Table \ref{Table 2.} leads essentially to the same type of results.

\subsection{Effect of inflammation}
\label{inflammation}

The results of Theorem \ref{existence} motivates the analysis of the steady states in function of the degradation rate of monomers $d$. In particular we are interested in the inflammation processes that leads to the persistence of AD. In this context we analyze the bifurcation diagram for the concentration of interleukins at equilibrium $I^*$ depending on the degradation rate of monomers $d$ as the bifurcation parameter. The rest of the components of a steady state of Equation \eqref{ode1-u-m-M-I} are calculated according the system \eqref{ode1-u-m-M-I-est}.

We observe in Figure \ref{fig:bif-I-d} that for all $d>0$ the disease-free equilibrium is asymptotically stable. Moreover there exists a critical degradation rate of monomers $d_c\approx 0.4779\:(\textrm{months})^{-1}$, which we call the \textbf{critical degradation rate of persistence}, such that for $d<d_c$ there exists two positive steady states where the maximal one is asymptotically stable and the other one is linearly unstable. If $d>d_c$ then the disease-free steady state is the only equilibrium of the system.

\begin{figure}[!ht]
\centering
\includegraphics[scale=0.5]{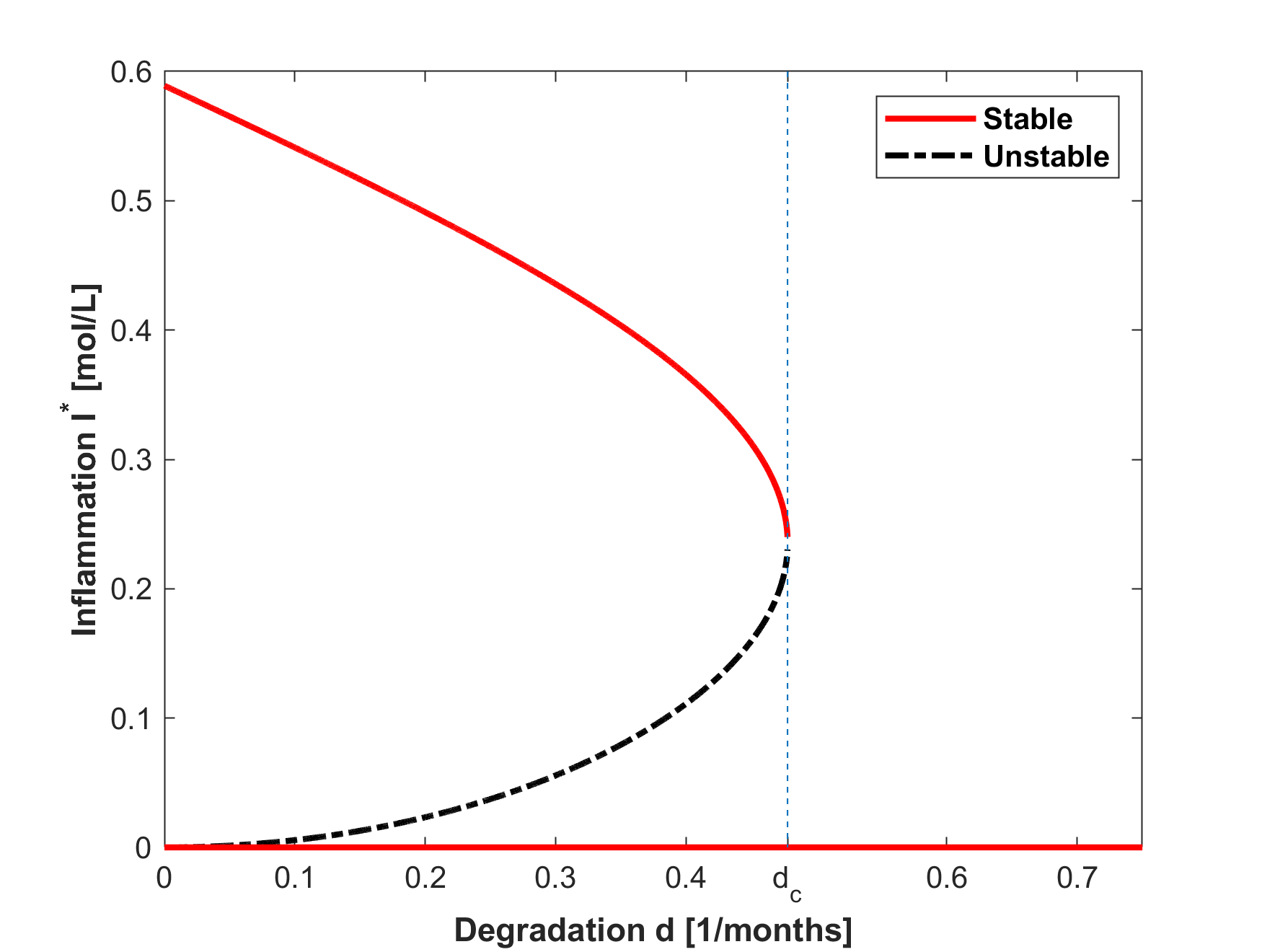}
\caption{Bifurcation diagram of the steady states for the concentration of interleukins $I^*$ in terms of the degradation rate of monomers $d$ with the parameters of Table \ref{Table 2.}. The disease-free equilibrium exists for all values of $d>0$ and it is stable. For $d<d_c$ we have other two non-trivial equilibria, where the maximal one is stable and the other one is unstable. For $d>d_c$ we get no positive steady states.}
\label{fig:bif-I-d}
\end{figure}

From the bifurcation diagram of Figure \ref{fig:bif-I-d}, we observe the importance of the degradation rate of monomers in determining the existence of steady states where AD persists. We also observe that for a small degradation rate $d$, the concentration of interleukins at equilibrium $I^*$ is large. 

The bifurcation analysis is quite challenging even for the simplified version of the model. Thus we proceed to numerical simulations in the next section in order to show the asymptotic behavior of solutions of system \eqref{ode1-u-m-M-I} under different degradation rates of monomers $d$ and initial data. In particular, we choose the initial values given in Table \ref{Table 3.}.

\begin{table}[!ht]
    \centering
    \caption{Initial data for the numerical simulations of Equation \eqref{ode1-u-m-M-I}.}
    \begin{tabular}{|l|l|l|l|}
    \hline
        Parameter &  Value & Units & Description\\ \hline
        $u_0$ & $10^{-4}$ & $L\,(\text{mol})^{-1}$ & Concentration of free oligomers\\ \hline
        ${u_p}_0$ & $0$ & $L\,(\text{mol})^{-1}$ & Concentration of oligomers in the amyloid plaques\\ \hline
        $m_0$ & $10^{-3}$ & $L\,(\text{mol})^{-1}$ & Concentration of monomers\\ \hline
        $M_0$ & $1$ & $L\,(\text{mol})^{-1}$ & Concentration of microgial cells\\ \hline
        $I_0$ & Variable & $L\,(\text{mol})^{-1}$ & Concentration of interleukins\\ \hline
    \end{tabular}
    \label{Table 3.}
\end{table}

This means that we study system \eqref{ode1-u-m-M-I} under a small initial concentration of monomers and free oligomers. We also consider that oligomers in the amyloid plaques are initially absent, while microglial cells are already developed. We vary the initial concentration of interleukins $I_0$ and the degration rate of monomers $d$ to study the asymptotic behavior of Equation \eqref{ode1-u-m-M-I}. 

In Figure \ref{fig:critical} we present the possible asymptotic behaviors of system \eqref{ode1-u-m-M-I} in terms of the degradation rate of monomers $d$ and the initial inflammation $I_0$ with the parameters in Table \ref{Table 2.} and initial data in Table \ref{Table 3.}.

\begin{figure}[!ht]
\centering
\includegraphics[scale=0.5]{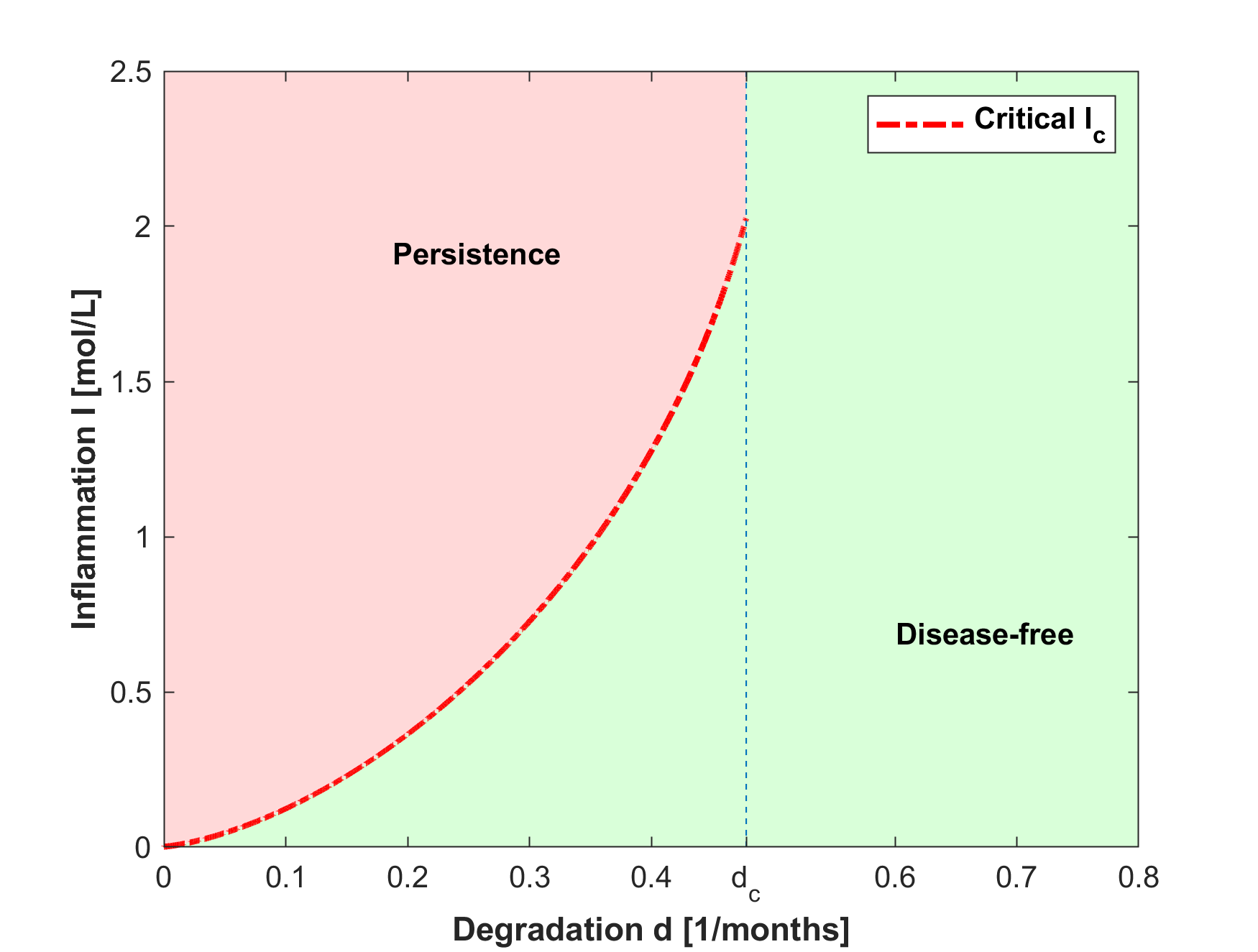}
\caption{Asymptotic behavior of solutions in terms of degradation rate of monomers $d$ and the initial inflammation $I_0$. For the parameter values from Table \ref{Table 2.} and initial data from Table \ref{Table 3.} we get the critical threshold of inflammation $I_c$ and the critical degradation rate $d_c$. For $d<d_c$ we get that AD persists for $I_0>I_c$ and does not persist if $I_0<I_c$. If $d>d_c$ the disease does not persist.}
\label{fig:critical}
\end{figure}

In particular we observe a phenomenon of hysteresis for $d<d_c$, where $d_c$ is the critical degradation rate in Figure \ref{fig:bif-I-d}, which implies the existence of a critical threshold value for the inflammation $I_c>0$ (depending on the rest of parameters and the initial data), that determine if AD persists or not. We observe in Figure \ref{fig:critical} that for degradation rates of monomers satisfying $d<d_c$, solutions of Equation \eqref{ode1-u-m-M-I} converge to the disease-free equilibrium when $I_0<I_c$ and converge to the positive stable equilibrium when $I_0>I_c$.

Moreover, for small values of $d$ a small initial concentration of interleukins $I_0$ suffices for the persistence of AD, while for values close the critical degradation rate of persistence $d_c$, a higher initial concentration of $I_0$ is needed. Furthermore, for $d<d_c$ most of solutions converge either to the disease-free equilibrium or the stable positive equilibrium. This global stability result is to be proven in a future work. When $d>d_c$, in absence of positive steady states, we conjecture that all the solutions of system \eqref{ode1-u-m-M-I} converge to the disease-free equilibrium.    

Next, we show some numerical simulations of solutions of the simplified system \eqref{ode1-u-m-M-I} in order to illustrate the effects of hysteresis and inflammation processes in the convergence to a steady state.

For a small degradation rate, $d=0.15\:(\textrm{months})^{-1}$ we observe from Figure \ref{fig:bif-I-d} that we have three steady states and by choosing $I_0=0.15\:\textrm{mol/L}$ we observe in Figure \ref{Simulation_1} that the solution converges to the disease-free equilibrium. In this example the concentration of interleukins is decreasing and the threshold of inflammation is not reached. Moreover, the concentrations of monomers increases until it reaches the maximum value and eventually decreases and the concentrations of free oligomers and oligomers in the amyloid plaques remain relatively low.

\begin{figure}[!ht]
    \centering
    \includegraphics[scale=0.12]{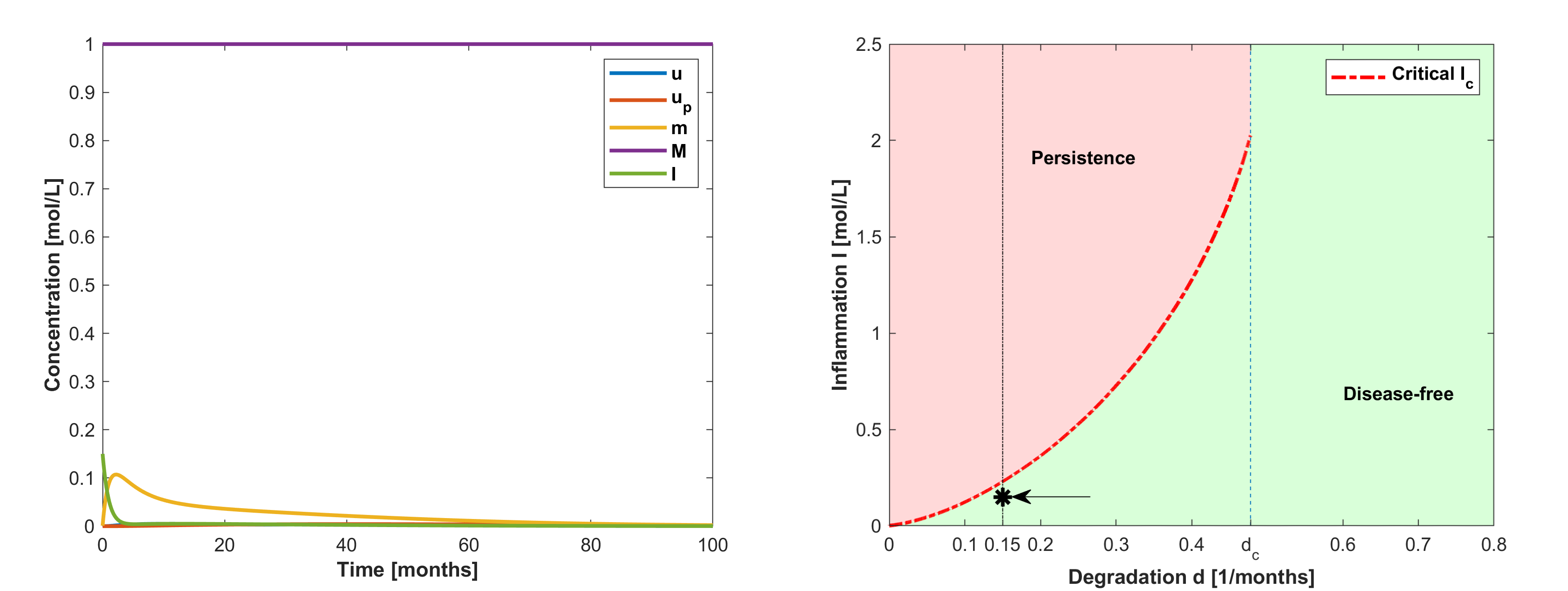}
    \caption{Example 1. (Left) Numerical solution of system \eqref{ode1-u-m-M-I} with $d=0.15\:(\textrm{months})^{-1}$ and $I_0=0.15\:\textrm{mol/L}$. The parameters correspond to those in Table \ref{Table 2.} and the initial data in Table \ref{Table 3.}. (Right) Asymptotic behavior in terms of degradation rate of monomers $d$ and the initial inflammation $I_0$. The value of $I_0$ is indicated with an arrow and $d$ by a vertical line.}
    \label{Simulation_1}
\end{figure}

If we increase the value of initial inflammation to $I_0=0.4\:\textrm{mol/L}$ in Figure \ref{Simulation_2} the solution converges to the stable positive steady state, since the critical threshold value $I_c$ is less than $I_0$. In this example the concentration of free oligomers and oligomers in the amyloid plaques are increasing towards the corresponding values of equilibrium. Inflammation is initially decreasing until it reaches the minimum value and eventually increases towards the equilibrium value, while the concentration of monomers has increasing and decreasing phases due the effect of stress mechanisms, nucleation and degradation.

\begin{figure}[!ht]
    \centering
    \includegraphics[scale=0.12]{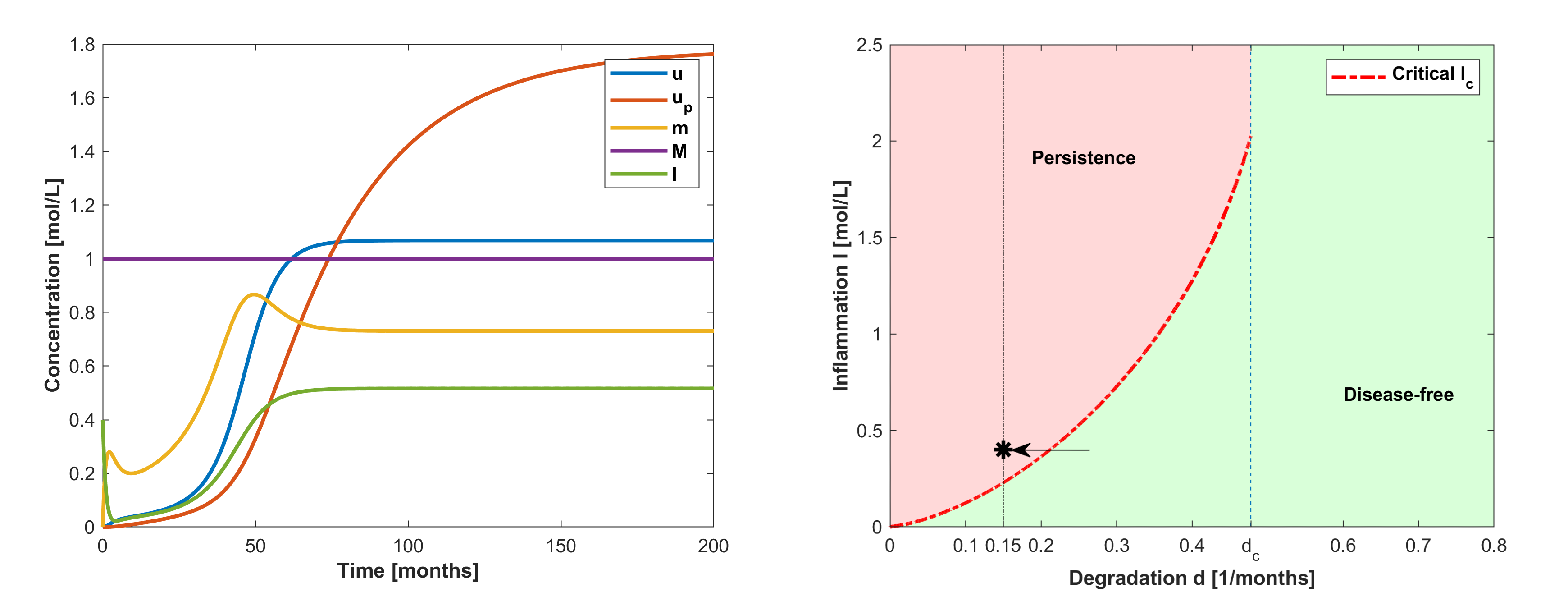}
    \caption{Example 2. (Left) Numerical solution of system \eqref{ode1-u-m-M-I} with $d=0.15\:(\textrm{months})^{-1}$ and $I_0=0.4\:\textrm{mol/L}$. The parameters correspond to those in Table \ref{Table 2.} and the initial data in Table \ref{Table 3.}. (Right) Asymptotic behavior in terms of degradation rate of monomers $d$ and the initial inflammation $I_0$. The value of $I_0$ is indicated with an arrow and $d$ by a vertical line.}
    \label{Simulation_2}
\end{figure}

In a similar way for a larger degradation rate $d=0.35\:(\textrm{months})^{-1}$, we have also three steady states according to Figure \eqref{fig:bif-I-d}. For $I_0=0.97\:\textrm{mol/L}$ we observe in Figure \ref{Simulation_3} that the solution converges to the disease-free equilibrium. In this example the concentration of interleukins is eventually decreasing, since the threshold of inflammation is not reached. Moreover the concentration of monomers, free oligomers and oligomers in the amyloid plaques increase until they reach their corresponding maximum values and eventually decrease. In particular the maxima are higher compared to those observed in Figure \ref{Simulation_1}.

\begin{figure}[!ht]
    \centering
    \includegraphics[scale=0.12]{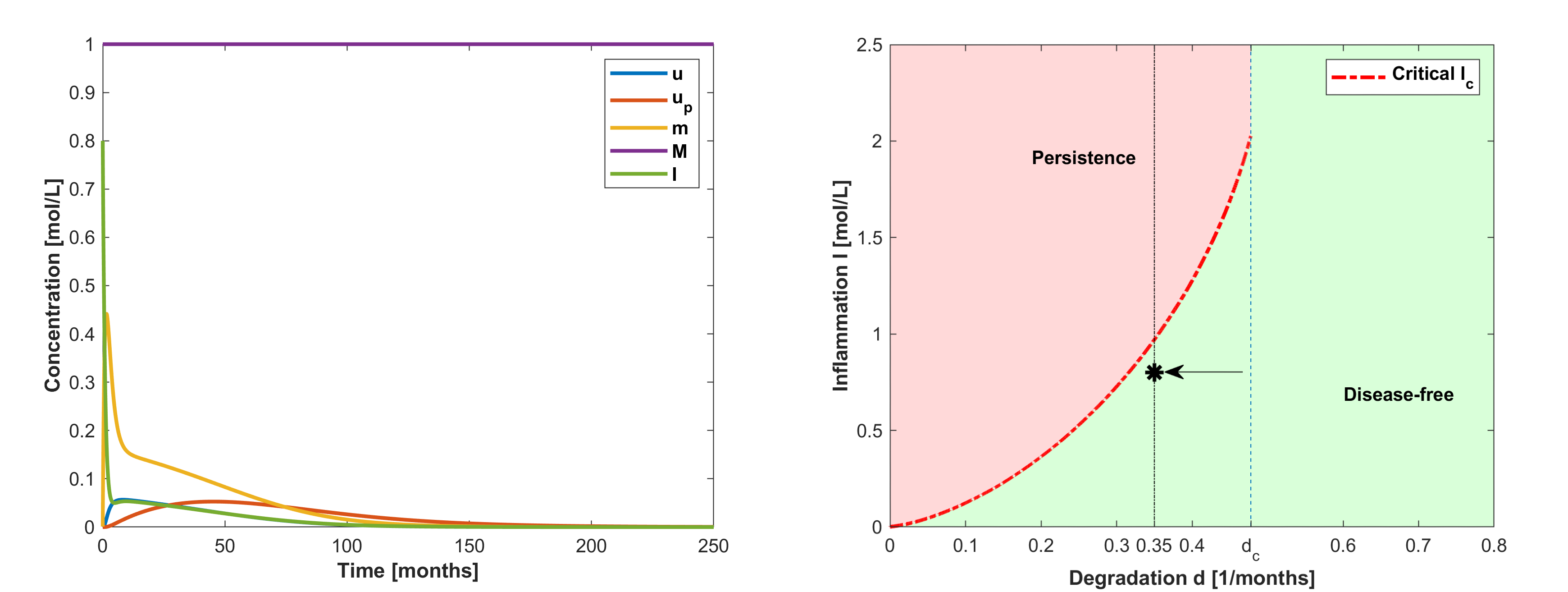}
    \caption{Example 3. (Left) Numerical solution of system \eqref{ode1-u-m-M-I} with $d=0.35\:(\textrm{months})^{-1}$ and $I_0=0.8\:\textrm{mol/L}$. The parameters correspond to those in Table \ref{Table 2.} and the initial data in Table \ref{Table 3.}. (Right) Asymptotic behavior in terms of degradation rate of monomers $d$ and the initial inflammation $I_0$. The value of $I_0$ is indicated with an arrow and $d$ by a vertical line.}
    \label{Simulation_3}
\end{figure}

For $I_0=0.98\:\textrm{mol/L}$ the solution converges to the positive stable steady state in Figure \ref{Simulation_4}, leading to the persistence of AD since the critical threshold value $I_c$ is less than $I_0$. Similarly to Figure \ref{Simulation_2}, the concentration of free oligomers and oligomers in the amyloid plaques are increasing towards the corresponding values of equilibrium. Inflammation is initially decreasing until it reaches the minimum value and eventually increases towards the equilibrium value, while the concentration of monomers has increasing and decreasing phases due the effect of stress mechanisms, nucleation and degradation. Moreover we observe that equilibrium values are lower to those observed in Figure \ref{Simulation_2} since the degradation rate of monomers is higher.

\begin{figure}[h!]
    \centering
    \includegraphics[scale=0.12]{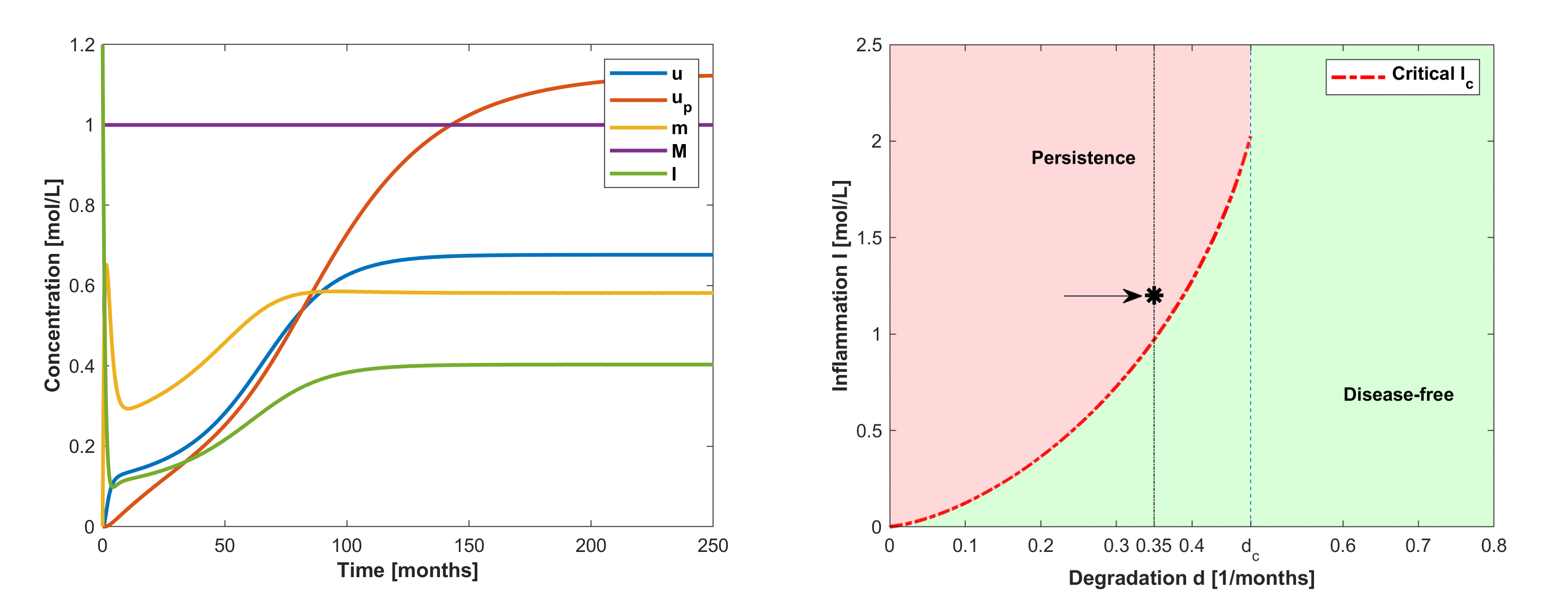}
    \caption{Example 4. (Left) Numerical solution of system \eqref{ode1-u-m-M-I} with $d=0.35\:(\textrm{months})^{-1}$ and $I_0=1.2\:\textrm{mol/L}$. The parameters correspond to those in Table \ref{Table 2.} and the initial data in Table \ref{Table 3.}. (Right) Asymptotic behavior in terms of degradation rate of monomers $d$ and the initial inflammation $I_0$. The value of $I_0$ is indicated with an arrow and $d$ by a vertical line.}
    \label{Simulation_4}
\end{figure}

Finally, for $d=0.55\:(\textrm{months})^{-1}$ we get only the trivial steady state according to Figure \ref{fig:bif-I-d}, so that for $I_0=2\:\textrm{mol/L}$ the solution converges to the disease-free equilibrium as we see in Figure \ref{Simulation_5}. The behavior of concentrations is similar to that in Figure \ref{Simulation_3}.

\begin{figure}[!ht]
\centering
\includegraphics[scale=0.12]{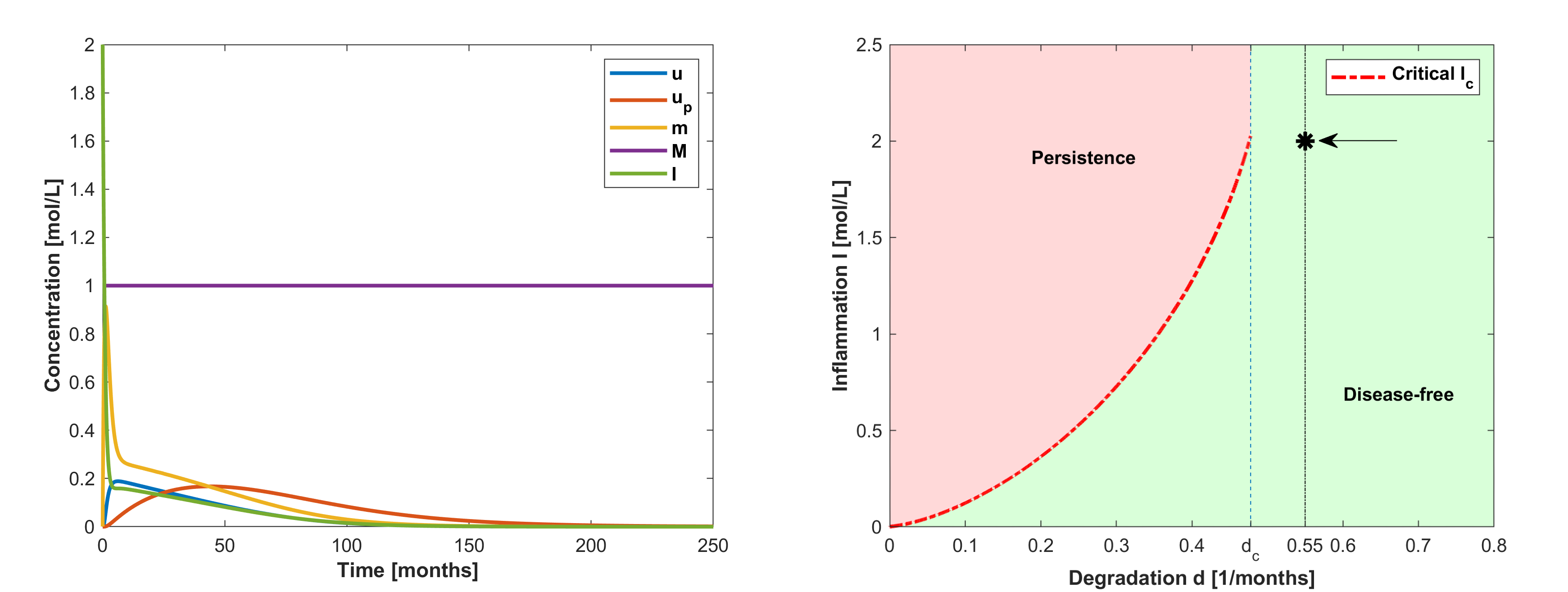}
\caption{Example 5. (Left) Numerical simulations of system \eqref{ode1-u-m-M-I} for the parameters in Table \ref{Table 2.} and the initial data in Table \ref{Table 3.} with $d=0.55\:(\textrm{months})^{-1}$ and $I_0=2\:\textrm{mol/L}$. (Right) Asymptotic behavior in terms of degradation rate of monomers $d$ and the initial inflammation $I_0$. The value of $I_0$ is indicated with an arrow and $d$ by a vertical line.}
\label{Simulation_5}
\end{figure}

In this bi-stable case, the solutions of system \eqref{ode1-u-m-M-I} converge to the positive stationary equilibrium if the initial concentrations of interleukins, monomers or free oligomers are sufficiently large. The phenomenon of hysteresis could indicate that AD can be initiated by the inflammation.

\subsection{Effect of monomer concentration}

Similarly to the analysis of inflammation in the persistence of AD, we study the effect of the initial concentration of monomers. In this context we present some numerical simulations to illustrate the same hysteresis phenomenon with respect to the initial concentration of monomers. We choose the initial values given in Table \ref{Table 4.}.

\begin{table}[!ht]
\centering
    \caption{Initial data for the numerical simulations of Equation \eqref{ode1-u-m-M-I}.}
    \begin{tabular}{|l|l|l|l|}
    \hline
        Parameter &  Value & Units & Description\\ \hline
        $u_0$ & 0 & $L\,(\text{mol})^{-1}$ & Concentration of free oligomers\\ \hline
        ${u_p}_0$ & $0$ & $L\,(\text{mol})^{-1}$ & Concentration of oligomers in the plaques\\ \hline
        $m_0$ & Variable & $L\,(\text{mol})^{-1}$ & Concentration of monomers\\ \hline
        $M_0$ & $1$ & $L\,(\text{mol})^{-1}$ & Concentration of microgial cells\\ \hline
        $I_0$ & $0$ & $L\,(\text{mol})^{-1}$ & Concentration of interleukins\\ \hline
    \end{tabular}
    \label{Table 4.}
\end{table}

This means that we study system \eqref{ode1-u-m-M-I} under a given concentration of oligomers while free oligomer, oligomers in the plaques and interleukins are initially absent. As in the previous analysis of Subsection \ref{inflammation} we assume that microglial cells are already developed. We vary the initial concentration of interleukins $m_0$ and the rate of monomers $d$ to show the asymptotic behavior of Equation \eqref{ode1-u-m-M-I}. We remark that similar results are obtained if we take a positive initial concentration of free oligomers and monomers are initially absent.

In Figure \ref{fig:critical_m} we present the possible asymptotic behaviors of system \eqref{ode1-u-m-M-I} in terms of the degradation rate of monomers $d$ and the initial concentration of monomers $m_0$, with the parameters of Table \ref{Table 2.} and initial data of Table \ref{Table 4.}, following the same analysis presented in Figure \ref{fig:critical}.

\begin{figure}[!ht]
\centering
\includegraphics[scale=0.5]{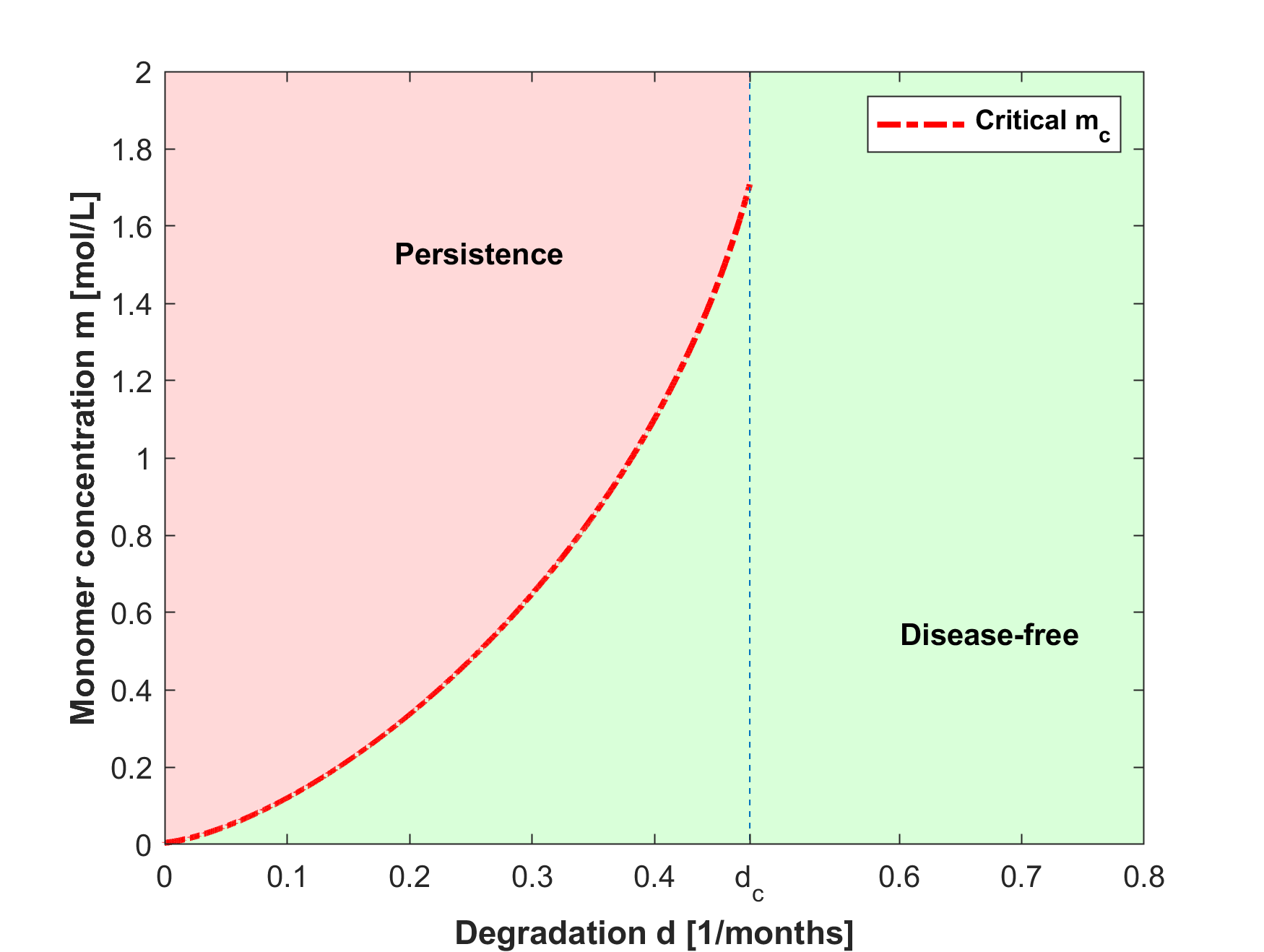}
\caption{Asymptotic behavior of solutions in terms of degradation rate of monomers $d$ and the initial concentration of monomers $m_0$. For the parameter values of Table \ref{Table 2.} and initial data of Table \ref{Table 4.} we get the critical threshold of monomer concentration $m_c$ and the critical degradation rate $d_c$. For $d<d_c$ we get that AD persists for $m_0>m_c$ and does not persist if $m_0<m_c$. If $d>d_c$ the disease does not persist.}
\label{fig:critical_m}
\end{figure}

Similarly to the previous Subsection \ref{inflammation}, we observe the same phenomenon of hysteresis for $d<d_c$, where $d_c$ is the critical degradation rate in Figure \ref{fig:bif-I-d}, which implies the existence of the respective critical threshold value for the initial concentration of monomers $m_c>0$ (depending on the rest of parameters and the initial data), that determine if AD persists or not. We observe in Figure \ref{fig:critical_m} that for degradation rates of monomers satisfying $d<d_c$, solutions of Equation \eqref{ode1-u-m-M-I} converge to the disease-free equilibrium when $m_0<m_c$ and converge to the positive stable equilibrium when $m_0>m_c$.

For $d=0.35\:(\textrm{months})^{-1}$ and $m_0=0.7\:\textrm{mol/L}$, we observe in Figure \ref{Simulation_6} that the solution converges to the disease-free equilibrium. In this example the concentration of monomers is decreasing (contrary to the case of the interleukins in the previous subsection), due to its intrinsic degradation rate and the formation of free oligomers. Moreover, the concentrations of free oligomers, oligomers in the amyloid plaques and interleukins increase until they reach their corresponding maximum values and eventually decrease in the same way as in the previous examples.

\begin{figure}[!ht]
    \centering
    \includegraphics[scale=0.12]{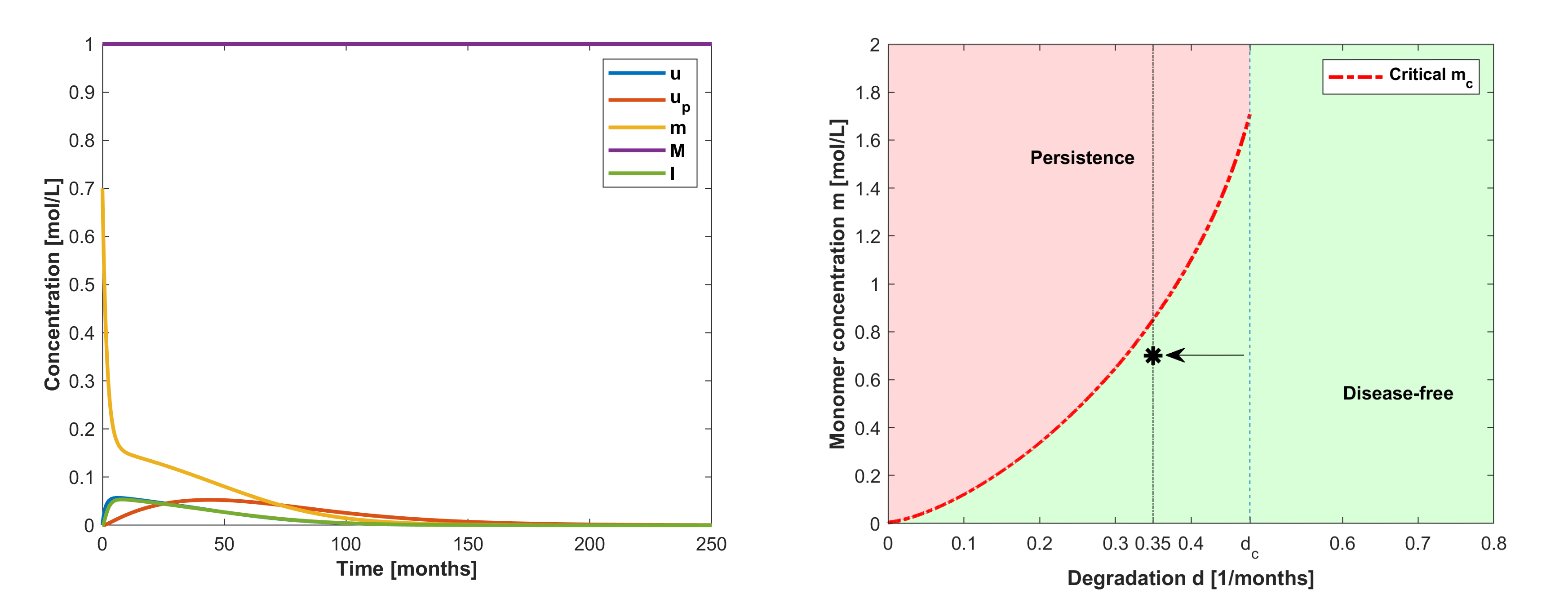}
    \caption{Example 6. (Left) Numerical solution of system \eqref{ode1-u-m-M-I} with $d=0.35\:(\textrm{months})^{-1}$ and $m_0=0.7\:\textrm{mol/L}$. The parameters correspond to those in Table \ref{Table 2.} and the initial data in Table \ref{Table 4.}. (Right) Asymptotic behavior in terms of degradation rate of monomers $d$ and the initial concentration of monomers $m_0$. The value of $m_0$ is indicated with an arrow and $d$ by a vertical line.}
    \label{Simulation_6}
\end{figure}

For the same value of the degradation rate $d$ and $I_0=1\:\textrm{mol/L}$ we observe in Figure \ref{Simulation_7} that the solution converges to the positive stable steady state, since the critical threshold value $m_c$ is less than $m_0$. In this example the concentration of free oligomers, oligomers in the amyloid plaques and interleukins are increasing towards the corresponding values of equilibrium. The monomers is initially decreasing until it reaches the minimum value and eventually increases towards the equilibrium value.

\begin{figure}[!ht]
    \centering
    \includegraphics[scale=0.12]{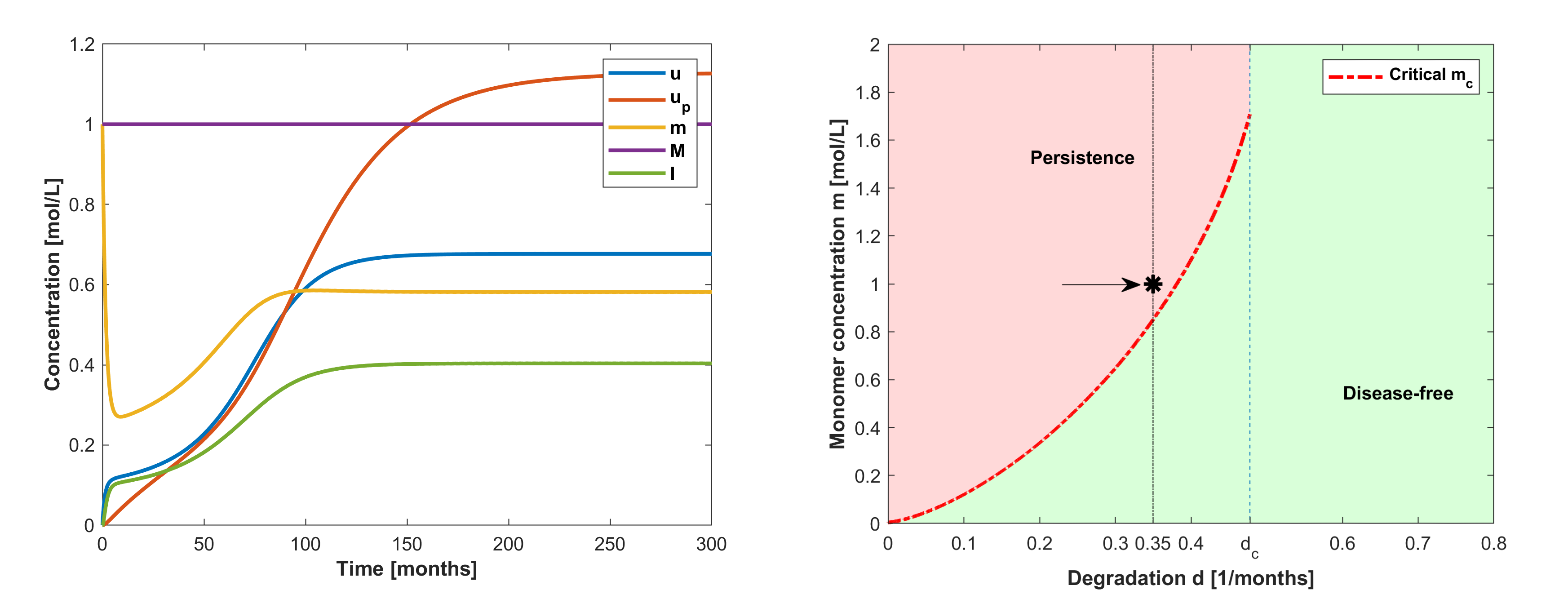}
    \caption{Example 7. (Left) Numerical solution of system \eqref{ode1-u-m-M-I} with $d=0.35\:(\textrm{months})^{-1}$ and $m_0=1\:\textrm{mol/L}$. The parameters correspond to those in Table \ref{Table 2.} and the initial data in Table \ref{Table 4.}. (Right) Asymptotic behavior in terms of degradation rate of monomers $d$ and the initial concentration of monomers $m_0$. The value of $m_0$ is indicated with an arrow and $d$ by a vertical line.}
    \label{Simulation_7}
\end{figure}

From Figures \ref{Simulation_6} and \ref{Simulation_7}, we observe essentially the same phenomenon of hysteresis and asymptotic behavior as for the simulations in Subsection \ref{inflammation}.
 
\section{Discussion and perspectives}
From the previous numerical simulations of the bi-monomeric model \eqref{ode1-u-m-M-I} in Section \ref{numericalsim}, and even if it corresponds to a simplified version of the original model, we already get a first qualitative approach in understanding the influence of inflammation and the degradation rates in the persistence of AD through a phenomenon of hysteresis, which determines the asymptotic behavior of solutions of system \eqref{ode1-u-m-M-I} through a critical threshold for the inflammation in terms of the parameters and the initial data. This qualitative analysis suggest that AD may be triggered by an initial high concentration of interleukins and its progression could be mitigated if an efficacious anti-inflammatory treatment would be applied in an early stage of disease, as it is suggested in \cite{imbimbo2010nsaids}. Furthermore, an interesting approach might be the study the effective times of applying anti-inflammatory doses in order to complement the stress mechanism given by the UPR in lowering the production of A$\beta$-monomers and not interfering with microglia activation cycles that counteracts the excess of toxic amyloid.

In this context, a possible extension of this study relies on modeling of such treatments via an impulsive differential equation for the concentration of interleukins $I$ (see \cite{lakshmikantham1989theory,samoilenko1995impulsive} for a reference on this type of differential equations). This could lead to interesting optimal control problems in order to optimize both time and quantity of dose provided to mitigate AD, inspired in the works of Hu et al. \cite{hu2022finite}. Moreover, another important extension to the presented model is the incorporation of cell destruction due to the accumulation of oligomers in the amyloid plaques. In particular, the stress function \eqref{stress} will also depend on neural population.

Concerning dynamics of the full model incorporating the spatial dependence, the chemotaxis of microglial cells and the whole polymerization process of proto-oligomers are far from being fully understood. For the whole and complete model, we expect a similar phenomenon of hysteresis to the one observed in the spatial-homogeneous simplified model, though the analysis to prove existence of steady states becomes way more challenging.

\section*{Acknowledgments}
This project has received support from Agence National de la Recherche PrionDiff Project-ANR-21-CE15-0011.

\bibliography{Alzheimer.bib}

\begin{thebibliography}{10}

\bibitem{haass2007soluble}
Haass C, Selkoe DJ.
\newblock Soluble protein oligomers in neurodegeneration: lessons from the
  Alzheimer's amyloid $\beta$-peptide.
\newblock Nature reviews Molecular cell biology. 2007;8(2):101--112.

\bibitem{Sakono2010}
Sakono M, Zako T.
\newblock {Amyloid oligomers: formation and toxicity of A$\beta$ oligomers}.
\newblock FEBS J. 2010;277(6):1348--1358.

\bibitem{sengupta2016role}
Sengupta U, Nilson AN, Kayed R.
\newblock The role of amyloid-$\beta$ oligomers in toxicity, propagation, and
  immunotherapy.
\newblock EBioMedicine. 2016;6:42--49.

\bibitem{soto2003unfolding}
Soto C.
\newblock Unfolding the role of protein misfolding in neurodegenerative
  diseases.
\newblock Nature Reviews Neuroscience. 2003;4(1):49--60.

\bibitem{ciuperca2019alzheimer}
Ciuperca IS, Dumont M, Lakmeche A, Mazzocco P, Pujo-Menjouet L, Rezaei H,
  et~al.
\newblock Alzheimer's disease and prion: An in vitro mathematical model.
\newblock Discrete \& Continuous Dynamical Systems-B. 2019;24(10):5225.

\bibitem{matthaus2006diffusion}
Matth{\"a}us F.
\newblock Diffusion versus network models as descriptions for the spread of
  prion diseases in the brain.
\newblock Journal of theoretical biology. 2006;240(1):104--113.

\bibitem{matthaeus2009spread}
Matthaeus F.
\newblock The spread of prion diseases in the brain models of reaction and
  transport networks.
\newblock Journal of Biological Systems. 2009;17(04):623--641.

\bibitem{Bertsch2016}
Bertsch M, Franchi B, Marcello N, Tesi MC, Tosin A.
\newblock {Alzheimer's disease: a mathematical model for onset and
  progression}.
\newblock Math Med Biol. 2016;~:~.

\bibitem{andrade2020modeling}
Andrade-Restrepo M, Lemarre P, Pujo-Menjouet L, Tine LM, Ciuperca SI.
\newblock Modeling the spatial propagation of A$\beta$ oligomers in
  Alzheimer’s Disease.
\newblock ESAIM: Proceedings and Surveys. 2020;67:30--45.

\bibitem{hu2022finite}
Hu J, Zhang Q, Meyer-Baese A, Ye M.
\newblock Finite-time stability and optimal control of a stochastic
  reaction-diffusion model for Alzheimer’s disease with impulse and
  time-varying delay.
\newblock Applied Mathematical Modelling. 2022;102:511--539.

\bibitem{murphy2000probing}
Murphy RM, Pallitto MM.
\newblock Probing the kinetics of $\beta$-amyloid self-association.
\newblock Journal of structural biology. 2000;130(2-3):109--122.

\bibitem{nag2011nature}
Nag S, Sarkar B, Bandyopadhyay A, Sahoo B, Sreenivasan VK, Kombrabail M, et~al.
\newblock Nature of the amyloid-$\beta$ monomer and the monomer-oligomer
  equilibrium.
\newblock Journal of Biological Chemistry. 2011;286(16):13827--13833.

\bibitem{forloni2018alzheimer}
Forloni G, Balducci C.
\newblock Alzheimer’s disease, oligomers, and inflammation.
\newblock Journal of Alzheimer's Disease. 2018;62(3):1261--1276.

\bibitem{kinney2018inflammation}
Kinney JW, Bemiller SM, Murtishaw AS, Leisgang AM, Salazar AM, Lamb BT.
\newblock Inflammation as a central mechanism in Alzheimer's disease.
\newblock Alzheimer's \& Dementia: Translational Research \& Clinical
  Interventions. 2018;4:575--590.

\bibitem{hansen2018microglia}
Hansen DV, Hanson JE, Sheng M.
\newblock Microglia in Alzheimer’s disease.
\newblock Journal of Cell Biology. 2018;217(2):459--472.

\bibitem{mazaheri2017trem}
Mazaheri F, Snaidero N, Kleinberger G, Madore C, Daria A, Werner G, et~al.
\newblock TREM 2 deficiency impairs chemotaxis and microglial responses to
  neuronal injury.
\newblock EMBO reports. 2017;18(7):1186--1198.

\bibitem{ransohoff2016polarizing}
Ransohoff RM.
\newblock A polarizing question: do M1 and M2 microglia exist?
\newblock Nature neuroscience. 2016;19(8):987--991.

\bibitem{sarlus2017microglia}
Sarlus H, Heneka MT, et~al.
\newblock Microglia in Alzheimer’s disease.
\newblock The Journal of clinical investigation. 2017;127(9):3240--3249.

\bibitem{rivers2020anti}
Rivers-Auty J, Mather AE, Peters R, Lawrence CB, Brough D.
\newblock Anti-inflammatories in Alzheimer’s disease—potential therapy or
  spurious correlate?
\newblock Brain communications. 2020;2(2):fcaa109.

\bibitem{adapt2008cognitive}
Group AR, et~al.
\newblock Cognitive function over time in the Alzheimer's Disease
  Anti-inflammatory Prevention Trial (ADAPT): results of a randomized,
  controlled trial of naproxen and celecoxib.
\newblock Archives of neurology. 2008;65(7):896.

\bibitem{adapt2007naproxen}
Group AR, et~al.
\newblock Naproxen and celecoxib do not prevent AD in early results from a
  randomized controlled trial.
\newblock Neurology. 2007;68(21):1800--1808.

\bibitem{ozben2019neuro}
Ozben T, Ozben S.
\newblock Neuro-inflammation and anti-inflammatory treatment options for
  Alzheimer's disease.
\newblock Clinical biochemistry. 2019;72:87--89.

\bibitem{ali2019recommendations}
Ali MM, Ghouri RG, Ans AH, Akbar A, Toheed A.
\newblock Recommendations for anti-inflammatory treatments in Alzheimer’s
  disease: a comprehensive review of the literature.
\newblock Cureus. 2019;11(5).

\bibitem{imbimbo2010nsaids}
Imbimbo BP, Solfrizzi V, Panza F.
\newblock Are NSAIDs useful to treat Alzheimer's disease or mild cognitive
  impairment?
\newblock Frontiers in aging neuroscience. 2010; p.~19.

\bibitem{huang2020clinical}
Huang LK, Chao SP, Hu CJ.
\newblock Clinical trials of new drugs for Alzheimer disease.
\newblock Journal of biomedical science. 2020;27(1):1--13.

\bibitem{zhu2018can}
Zhu M, Wang X, Sun L, Schultzberg M, Hjorth E.
\newblock Can inflammation be resolved in Alzheimer’s disease?
\newblock Therapeutic advances in neurological disorders.
  2018;11:1756286418791107.

\bibitem{haraux2016simple}
Haraux A.
\newblock A simple characterization of positivity preserving semi-linear
  parabolic systems.
\newblock arXiv preprint arXiv:161009909. 2016;.

\bibitem{lakshmikantham1989theory}
Lakshmikantham V, Simeonov PS, et~al.
\newblock Theory of impulsive differential equations. vol.~6.
\newblock World scientific; 1989.

\bibitem{samoilenko1995impulsive}
Samoilenko AM, Perestyuk N.
\newblock Impulsive differential equations.
\newblock World Scientific; 1995.

\end{thebibliography}
\bibliographystyle{plos2015}

\end{document}